\documentclass{llncs}

\usepackage{amsmath,amsfonts,amssymb}

\usepackage{graphicx,color}
\usepackage{boxedminipage}
\newtheorem{myclaim}{Claim}

\newenvironment{proofofclaim}[1][Proof]{\noindent\textbf{#1.} }{\ \rule{0.5em}{0.5em}}
\newcommand{\vd}{\textrm{vd}}
\newcommand{\ed}{\textrm{ed}}
\newcommand{\ea}{\textrm{ea}}

\DeclareMathOperator{\operatorClassNP}{NP}
\newcommand{\classNP}{\ensuremath{\operatorClassNP}}
\DeclareMathOperator{\operatorClassCoNP}{coNP}
\newcommand{\classCoNP}{\ensuremath{\operatorClassCoNP}}
\DeclareMathOperator{\operatorClassFPT}{FPT}
\newcommand{\classFPT}{\ensuremath{\operatorClassFPT}}
\DeclareMathOperator{\operatorClassW}{W}
\newcommand{\classW}[1]{\ensuremath{\operatorClassW[#1]}}

\begin{document}

\title{Graph Editing to a Given Degree Sequence\thanks{The research leading to these results has received funding from the European Research Council under the European Union's Seventh Framework Programme (FP/2007-2013)/ERC Grant Agreement n.~267959 and by the EPSRC Grant EP/K022660/1.}}

\author{Petr A. Golovach\inst{1}
\and
George B. Mertzios\inst{2}
}

\institute{Department of Informatics, University of Bergen, Norway.
\and
School of Engineering and Computing Sciences, Durham University, UK.
}

\maketitle

\begin{abstract}
 We investigate the parameterized complexity of  the graph editing problem called \textsc{Editing to a Graph with a Given Degree Sequence}, where the aim is to obtain a graph  with a given degree sequence $\sigma$ by at most $k$ vertex or edge deletions and edge additions. We show that the problem is \classW{1}-hard when parameterized by $k$ for any combination of the allowed editing operations. From the positive side, we show that the problem can be solved in time $2^{O(k(\Delta+k)^2)}n^2\log n$ for $n$-vertex graphs, where $\Delta=\max\sigma$, i.e., the problem is \classFPT{} when parameterized by $k+\Delta$.  We also show that \textsc{Editing to a Graph with a Given Degree Sequence} has a polynomial kernel when parameterized by $k+\Delta$ if  only edge additions are allowed, and  there is no polynomial kernel unless  $\classNP\subseteq\classCoNP/\text{\rm poly}$ for all other combinations of allowed editing operations.
\end{abstract}

\section{Introduction}
The aim of graph editing (or graph modification) problems is to modify a given graph  
by applying a bounded number of permitted operations in order to satisfy a certain
property. Typically, vertex deletions, edge deletions and edge additions are the considered as the permitted editing operations, but in some cases other operations like edge contractions and vertex additions are also permitted.

We are interested in graph editing problems, where the aim is to obtain a graph satisfying some given degree constraints. These problems usually turn out to be \classNP-hard (with rare exceptions). Hence, we are interested in the parameterized complexity of such problems. Before we state our results we briefly discuss the known related (parameterized) complexity results.

\medskip
\noindent
{\bf Related work.} 
The investigation of the parameterized complexity of such problems was initiated by 
Moser and Thilikos in~\cite{MoserT09} and Mathieson and Szeider~\cite{MathiesonS12}. 
In particular, Mathieson and Szeider~\cite{MathiesonS12} considered the \textsc{Degree Constraint Editing} problem that asks for a given graph $G$, nonnegative integers $d$ and $k$, and a function $\delta\colon V(G)\rightarrow 2^{\{0,\ldots,d\}}$, whether $G$ can be modified into a graph $G'$ such
that $d_{G'}(v)\in \delta(v)$ for each $v\in V(G')$, by using at most $k$ editing operations. They 
classified the (parameterized) complexity of the problem depending on the set of allowed editing operations.
In particular, they proved that if only edge deletions and additions are permitted, then the problem can be solved in polynomial time for the case where the set of feasible degrees $|\delta(v)|=1$ for $v\in V(G)$. 
Without this restriction on the size of the sets of feasible degrees, the problem is \classNP-hard even on subcubic planar graphs whenever only edge deletions are allowed~\cite{Cornuejols88} and whenever only edge additions are allowed~\cite{FroeseNN14}.
If vertex deletions can be used, then the problem becomes 
 \classNP-complete and \classW{1}-hard with parameter~$k$, even if the sets of feasible degrees have size oner~\cite{MathiesonS12}. 
 Mathieson and Szeider~\cite{MathiesonS12} showed that \textsc{Degree Constraint Editing} is  \classFPT\ when parameterized by $d+k$. They also proved that the problem has a polynomial kernel in  the case where only vertex and edge deletions are allowed and the sets of feasible degrees have size one. Further kernelization results were obtained by Froese, Nichterlein and Niedermeier~\cite{FroeseNN14}. In particular, they proved that the problem with the parameter $d$ admits a polynomial kernel if only edge additions are permitted.  They also complemented these results by showing that there is no polynomial kernel 
unless $\classNP\subseteq\classCoNP/\text{\rm poly}$ if only vertex or edge deletions are allowed. Golovach proved in~\cite{Golovach15} that, unless $\classNP\subseteq\classCoNP/\text{\rm poly}$, the problem does not admit a polynomial kernel when parameterized by $d+k$ if vertex deletion and edge addition are in the list of operations, even if the sets of feasible degrees have size one. 
The case where the input graph is planar was considered by Dabrowski et al. in~\cite{DabrowskiGHPT15}.
Golovach~\cite{Golovach14a} introduced a variant of \textsc{Degree Constraint Editing} in which, besides the degree restrictions, it is required that the graph obtained by editing should be connected. This variant for planar input graphs was also considered in~\cite{DabrowskiGHPT15}.

Froese, Nichterlein and Niedermeier~\cite{FroeseNN14} also considered the \textsc{$\Pi$-Degree Sequence Completion} problem which, 
given a graph $G$, a nonnegative integer $k$, and a property $\Pi$ of graph degree sequences, 
asks whether it is possible to obtain a graph $G\rq{}$ from $G$ by adding at most $k$ edges such that the degree sequence of $G\rq{}$ satisfies $\Pi$. 
They gave some conditions when the problem is \classFPT{}/admits a polynomial kernel when parameterized by $k$ and the maximum degree of $G$.
There are numerous results (see, e.g.,~\cite{BoeschST77,CaiY11,CyganMPPS14,DGHP14}) about the graph editing problem, where the aim is to obtain a (connected) graph whose vertices satisfy some parity restrictions on their degree. In particular, if the obtained graph is required to be a connected graph with vertices of even degree, we obtain the classical \textsc{Editing to Eulerian Graph} problem (see.~\cite{BoeschST77,DGHP14}). 

Another variant of graph editing with degree restrictions is the \textsc{Degree Anonymization} problem, motivated by some privacy and social networks applications. A graph $G$ is \emph{$h$-anonymous} for a positive integer $h$ if for any $v\in V(G)$, there is at least $h-1$ other vertices of the same degree. \textsc{Degree Anonymization} asks, given a graph $G$, a nonnegative $h$, and a positive integer $k$, whether it is possible to obtain an $h$-anonymous graph by at most $k$ editing operations. The investigation of the parameterized complexity of \textsc{Degree Anonymization} was initiated by Hartung et al.~\cite{HartungNNS15} and Bredereck et al.~\cite{BredereckHNW13} (see also~\cite{BredereckFHNNT14,HartungT15}). In particular,  Hartung et al.~\cite{HartungNNS15} considered the case where only edge additions are allowed. They proved that the problem is \classW{1}-hard when parameterized by $k$, but it becomes \classFPT{} and has a polynomial kernel when parameterized by the maximum degree $\Delta$ of an input graph. Bredereck et al. in~\cite{BredereckHNW13} considered vertex deletions. They proved that the problem is \classW{1}-hard when parameterized by $h+k$, but it is \classFPT{} when parameterized by $\Delta+h$ or by $\Delta+k$. Also the problem was investigated for the cases when vertex additions~\cite{BredereckFHNNT14} and edge contractions~\cite{HartungT15} are the editing operations.

\medskip
\noindent
{\bf Our results.}
Recall that the \emph{degree sequence} of a graph is the nonincreasing sequence of its vertex degrees. 
We consider the graph editing problem, where the aim is to obtain a graph with a given \emph{degree sequence} by using the operations \emph{vertex deletion}, \emph{edge deletion}, and \emph{edge addition}, 
denoted by $\vd$, $\ed$, and $\ea$, respectively. Formally, the problem is stated as follows. Let $S\subseteq\{\vd,\ed,\ea\}$.

\begin{center}
\begin{boxedminipage}{.99\textwidth}
\textsc{Editing to a Graph with a Given Degree Sequence}\\
\begin{tabular}{ r l }
\textit{~~Instance:} & A graph $G$, a nondecreasing sequence of nonnegative integers\\
& $\sigma$ and a nonnegative integer $k$.\\
\textit{Question:} & Is it possible to obtain a graph $G'$ with the degree sequence $\sigma$\\
& from $G$  by at most $k$ operations from $S$?\\
\end{tabular}
\end{boxedminipage}
\end{center}

It is worth highlighting here the difference between this problem and the \textsc{Editing to a Graph of Given Degrees} problem studied in~\cite{Golovach15}. 
In~\cite{Golovach15} a function $\delta:V(G)\rightarrow \{1,\ldots,d\}$ is given along with the input and, in the target graph~$G'$, every vertex~$v$ is required to have the \emph{specific} degree~$\delta(v)$. 
In contrast, in the \textsc{Editing to a Graph with a Given Degree Sequence}, only a degree sequence is given with the input and the requirement is that the target graph~$G'$ has this degree sequence, 
without specifying which specific vertex has which specific degree. 
To some extend, this problem can be seen as a generalization of the \textsc{Degree Anonymization} problem~\cite{HartungNNS15,BredereckHNW13,BredereckFHNNT14,HartungT15}, 
as one can specify (as a special case) the target degree sequence in such a way that every degree appears at least $h$ times in it.

In practical applications with respect to privacy and social networks, we might want to appropriately ``smoothen'' the degree sequence of a given graph in such a way that 
it becomes difficult to distinguish between two vertices with (initially) similar degrees. 
In such a setting, it does not seem very natural to specify in advance a \emph{specific} desired degree to every \emph{specific} vertex of the target graph. 
Furthermore, for anonymization purposes in the case of a social network, where the degree distribution often follows a so-called power law distribution~\cite{Barabasi99}, 
it seems more natural to identify a smaller number of vertices having all the same ``high'' degree, and a greater number of vertices having all the same ``small'' degree, 
in contrast to the more modest $h$-anonymization requirement where \emph{every} different degree must be shared among at least $h$ identified vertices in the target graph.

In Section~\ref{sec:defs}, we observe that  for any nonempty $S\subseteq\{\vd,\ed,\ea\}$, 
\textsc{Editing to a Graph with a Given Degree Sequence} is \classNP-complete and  \classW{1}-hard when parameterized by $k$. Therefore, we consider a stronger parameterization by $k+\Delta$, where $\Delta=\max\sigma$. In Section~\ref{sec:FPT}, we show that \textsc{Editing to a Graph with a Given Degree Sequence} is \classFPT\ when parameterized by $k+\Delta$. In fact, we obtain this result for the more general variant of the problem, where we ask whether we can obtain a graph $G\rq{}$ with the degree sequence $\sigma$ from an input graph $G$ by 
at most $k_{vd}$ vertex deletions, $k_{ed}$ edge deletions and $k_{ea}$ edge  additions. We show that the problem can be solved in time $2^{O(k(\Delta+k)^2)}n^2\log n$ for $n$-vertex graphs, where $k=k_{vd}+k_{ed}+k_{ea}$. The algorithm uses  the random separation techniques introduced by Cai, Chan and Chan~\cite{CaiCC06} (see also~\cite{AlonYZ95}). First, we construct a true biased Monte Carlo algorithm and then explain how it can be derandomized. 
In Section~\ref{sec:kernel}, we show that \textsc{Editing to a Graph with a Given Degree Sequence} has a polynomial kernel when parameterized by $k+\Delta$ if $S=\{ea\}$, but for all other nonempty $S\subseteq \{vd,ed,ea\}$, there is no polynomial kernel unless  $\classNP\subseteq\classCoNP/\text{\rm poly}$.

\section{Basic definitions and preliminaries}\label{sec:defs}

\noindent
{\bf Graphs.}
We consider only finite undirected graphs without loops or multiple
edges. The vertex set of a graph $G$ is denoted by $V(G)$ and  
the edge set  is denoted by $E(G)$.

For a set of vertices $U\subseteq V(G)$,
$G[U]$ denotes the subgraph of $G$ induced by $U$, and by $G-U$ we denote the graph obtained from $G$ by the removal of all the vertices of $U$, i.e., the subgraph of $G$ induced by $V(G)\setminus U$. If $U=\{u\}$, we write $G-u$ instead of $G-\{u\}$. 
Respectively, for a set of edges $L\subseteq E(G)$, $G[L]$ is a subgraph of $G$ induced by $L$, i.e, the vertex set of $G[L]$ is the set of vert ices of $G$ incident to the edges of $L$, and $L$ is the set of edges of $G[L]$.  
For a nonempty set $U$, $\binom{U}{2}$ is the set of unordered pairs of elements of $U$.
For a set of edges $L$,
by $G-L$ we denote the graph obtained from $G$ by the removal of all the edges of $L$.
Respectively, for $L\subseteq \binom{V(G)}{2}$, $G+L$ is the graph obtained from $G$ by the addition of the edges that are elements of $L$.
If $L=\{a\}$, then for simplicity, we write $G-a$ or $G+a$.

For a vertex $v$, we denote by $N_G(v)$ its
\emph{(open) neighborhood}, that is, the set of vertices which are adjacent to $v$, and for a set $U\subseteq V(G)$, $N_G(U)=(\cup_{v\in U}N_G(v))\setminus U$.
The \emph{closed neighborhood} $N_G[v]=N_G(v)\cup \{v\}$, and for a positive integer $r$, $N_G^r[v]$ is the set of vertices at distance at most $r$ from $v$. 
For a set $U\subseteq V(G)$ and a positive integer $r$, $N_G^r[U]=\cup_{v\in U}N_G^r[v]$.
The \emph{degree} of a vertex $v$ is denoted by $d_G(v)=|N_G(v)|$. The \emph{maximum degree} $\Delta(G)=\max\{d_G(v)\mid v\in V(G)\}$.

For a graph $G$, we denote by $\sigma(G)$ its degree sequence. Notice that $\sigma(G)$ can be represented 
by the vector $\delta(G)=(\delta_0,\ldots,\delta_{\Delta(G)})$, where 
$\delta_i=|\{v\in V(G)\mid d_G(v)=i\}|$ for $i\in\{0,\ldots,\Delta(G)\}$.
We call $\delta(G)$ the \emph{degree vector} of $G$. 
For a sequence $\sigma=(\sigma_1,\ldots,\sigma_n)$, $\delta(\sigma)=(\delta_0,\ldots,\delta_r)$,
where $r=\max\sigma$ and $\delta_i=|\{\sigma_j\mid \sigma_j=i\}|$ for $i\in\{0,\ldots,r\}$.
Clearly, $\delta(G)=\delta(\sigma(G))$, and
the degree vector can be easily constructed from the degree sequence and vice versa.
Slightly abusing notation, we write for two vectors of nonnegative integers, that $(\delta_0,\ldots,\delta_r)=(\delta_0',\ldots,\delta_{r'}')$ for $r\leq r'$ if $\delta_i=\delta_i'$ for $i\in\{0,\ldots,r\}$ and
$\delta_i'=0$ for $i\in\{r+1,\ldots,r'\}$.

\medskip
\noindent
{\bf Parameterized Complexity.}
Parameterized complexity is a two dimensional framework
for studying the computational complexity of a problem. One dimension is the input size
$n$ and another one is a parameter $k$. It is said that a problem is \emph{fixed parameter tractable} (or \classFPT), if it can be solved in time $f(k)\cdot n^{O(1)}$ for some function $f$.
A \emph{kernelization} for a parameterized problem is a polynomial algorithm that maps each instance $(x,k)$ with the input $x$ and the parameter $k$ to an instance $(x',k')$ such that i) $(x,k)$ is a YES-instance if and only if $(x',k')$ is a YES-instance of the problem, and ii) $|x'|+k'$ is bounded by $f(k)$ for a computable function $f$. 
The output $(x',k')$ is called a \emph{kernel}. The function $f$ is said to be a \emph{size} of a kernel. Respectively, a kernel is \emph{polynomial} if $f$ is polynomial. 
A parameterized problem is \classFPT{} if and only if it has a kernel, but it is widely believed that not all \classFPT{} problems have polynomial kernels. In particular, Bodlaender et al.~\cite{BodlaenderJK14} introduced techniques that allow to show that a parameterized problem has no polynomial kernel unless  $\classNP\subseteq\classCoNP/\text{\rm poly}$.
We refer to the recent books of Cygan et al.~\cite{CyganFKLMPPS15} and  Downey and Fellows~\cite{DowneyF13} for  detailed introductions  to parameterized complexity. 

\medskip
\noindent
{\bf Solutions of \textsc{Editing to a Graph with a Given Degree Sequence}.}
Let $(G,\sigma,k)$ be an instance of \textsc{Editing to a Graph of Given Degrees}.
Let $U\subset V(G)$, $D\subseteq E(G-U)$ and $A\subseteq \binom{V(G)\setminus U}{2}$.
 We say that $(U,D,A)$ is a \emph{solution} for $(G,\delta,d,k)$, if $|U|+|D|+|A|\leq k$, and the graph $G'=G-U-D+A$ has the degree sequence $\sigma$.  
We also say that $G'$ is obtained by editing with respect to $(U,D,A)$. 
 If $\vd$, $\ed$ or $\ea$ is not in $S$, then it is assumed that $U=\emptyset$, $D=\emptyset$ or $A=\emptyset$ respectively. If $S=\{\ed\}$, then instead of $(\emptyset,\emptyset,A)$ we simply write $A$.

We conclude this section by showing that \textsc{Editing to a Graph with a Given Degree Sequence} is hard when parameterized by $k$.

\begin{theorem}\label{thm:W-hard}
For any nonempty $S\subseteq\{\vd,\ed,\ea\}$, 
\textsc{Editing to a Graph with a Given Degree Sequence} is \classNP-complete and  \classW{1}-hard when parameterized by $k$.
\end{theorem}

\begin{proof}
Suppose that $\ed\in S$.
We reduce the \textsc{Clique} problem that asks for a graph $G$ and a positive integer $k$, whether $G$ has a clique of size $k$. This problem is known to be \classNP-complete~\cite{GareyJ79} and \classW{1}-hard when parameterized by $k$~\cite{Cai08} even if the input graph restricted to be  regular. Let $(G,k)$ be an instance of \textsc{Clique}, where $G$ is an $n$-vertex $d$-regular graph, $d\geq k-1$. Consider the sequence $\sigma=(\sigma_1,\ldots,\sigma_n)$, where
$$
\sigma_i=
\begin{cases}
d&\mbox{if }1\leq i\leq n-k,\\
d-(k-1)&\mbox{if } n-k+1\leq i\leq n.
\end{cases}
$$
Let $k'=k(k-1)/2$. We claim that $(G,k)$ is a yes-instance of \textsc{Clique} if and only if $(G,\sigma,k')$ is a yes-instance of \textsc{Editing to a Graph with a Given Degree Sequence}.
If $K$ is a clique of size $k$ in $G$, then the graph $G'$ obtained  from $G$ by the deletion of the $k'=k(k-1)/2$ edges of $D=E(G[K])$ has the degree sequence $\sigma$. Assume that 
$(U,D,A)$ is a solution of $(G,\sigma,k)$. Clearly, $U=\emptyset$ even if $\vd\in R$, because $\sigma$ contains $n$ elements. Since $\sum_{i=1}^n\sigma_i=dn-k(k-1)$, we have that $A=\emptyset$. It remains to notice that because in $G-D$ $k$ vertices have degree $d-(k-1)$, $G[D]$ is a compete graph with $k$ vertices, i.e., $G$ contains a clique of size $k$.

Suppose that $\ea\in S$. We reduce 
\textsc{Independent Set} problem that asks for a graph $G$ and a positive integer $k$, whether $G$ has an independent set of size $k$. Again, \textsc{Independent Set}  is \classNP-complete~\cite{GareyJ79} and \classW{1}-hard when parameterized by $k$~\cite{Cai08} even if the input graph restricted to be  regular. Let $(G,k)$ be an instance of \textsc{Independent Set}, where $G$ is an $n$-vertex $d$-regular graph and $k\leq n$. Consider the sequence $\sigma=(\sigma_1,\ldots,\sigma_n)$, where
$$
\sigma_i=
\begin{cases}
d+(k-1)&\mbox{if } 1\leq i\leq k,\\
d&\mbox{if } k+1\leq i\leq n.
\end{cases}
$$
Let $k'=k(k-1)/2$. Similarly to the case $\ed\in S$, we obtain that  $(G,k)$ is a yes-instance of \textsc{Independent Set} if and only if $(G,\sigma,k')$ is a yes-instance of \textsc{Editing to a Graph with a Given Degree Sequence}.

Finally, assume that $S=\{\vd\}$. We again reduce  the \textsc{Clique} problem for regular graphs. Let  $(G,k)$ be an instance of \textsc{Clique}, where $G$ is an $n$-vertex $d$-regular graph with $m$ edges.
 We assume without loss of generality that  $d-(k-1)\geq 3$. The graph $G'$ is constructed from $G$ by subdividing each edge of $G$, i.e., for each $xy\in E(G)$, we construct a new vertex $u$ and replace $xy$ by $xu$ and $yu$. Let $k'=k(k-1)/2$.
Consider the sequence $\sigma=(\sigma_1,\ldots,\sigma_p)$, where $p=n+m-k'$ and 
$$
\sigma_i=
\begin{cases}
d&\mbox{if } 1\leq i\leq n-k,\\
d-(k-1)&\mbox{if } n-k+1\leq i\leq n,\\
2&\mbox{if } n+1\leq i\leq p.
\end{cases}
$$
Again similarly to the case $\ed\in S$, we obtain that  $(G,k)$ is a yes-instance of \textsc{Clique} if and only if $(G',\sigma,k')$ is a yes-instance of \textsc{Editing to a Graph with a Given Degree Sequence}.\qed
\end{proof}

\section{\classFPT-algorithm for Editing to a Graph with a Given Degree Sequence}\label{sec:FPT}
In this section we show that \textsc{Editing to a Graph with a Given Degree Sequence} is \classFPT\ when parameterized by $k+\Delta$, where $\Delta=\max\sigma$. In fact, we obtain this result for the more general variant of the problem:

\begin{center}
\begin{boxedminipage}{.99\textwidth}
\textsc{Extended Editing to a Graph with a Given Degree Sequence}\\
\begin{tabular}{ r l }
\textit{~~Instance:} & A graph $G$, a nondecreasing sequence of nonnegative integers\\
& $\sigma$ and a nonnegative integers $k_{vd},k_{ed},k_{ea}$.\\
\textit{Question:} & Is it possible to obtain a graph $G'$ with $\sigma(G)=\sigma$ from $G$ by \\
& at most $k_{vd}$ vertex deletions, $k_{ed}$ edge deletions and $k_{ea}$ edge\\ 
& additions?\\
\end{tabular}
\end{boxedminipage}
\end{center}

\begin{theorem}\label{thm:FPT}
\textsc{Extended Editing to a Graph with a Given Degree Sequence} can be solved it time $2^{O(k(\Delta+k)^2)}n^2\log n$ for $n$-vertex graphs, where $\Delta=\max \sigma$ and $k=k_{vd}+k_{ed}+k_{ea}$.
\end{theorem}

\begin{proof}
First, we construct a randomized true biased Monte Carlo \classFPT-algorithm for \textsc{Extended Editing to a Graph with a Given Degree Sequence} parameterized by $k+d$ based on the random separation techniques introduced by Cai, Chan and Chan~\cite{CaiCC06} (see also~\cite{AlonYZ95}). Then we explain how this algorithm can be derandomized. 

Let $(G,S,k_{vd},k_{ed},k_{ea})$ be an instance of \textsc{Extended Editing to a Graph with a Given Degree Sequence}, $n=|V(G)|$. 

On the first stage of the algorithm we preprocess the instance to  get rid of vertices of high degree or solve the problem if we have a trivial no-instance by the following reduction rule.

\medskip
\noindent
{\bf Vertex deletion rule.} 
If $G$ has a vertex $v$ with $d_G(v)> \Delta+k_{vd}+k_{ed}$, then delete $v$ and set $k_{vd}=k_{vd}-1$. If $k_{vd}<0$, then stop and return a NO-answer.

\medskip
To show that the rule is \emph{safe}, i.e., by the application of the rule we either correctly solve the problem or obtain an equivalent instance, assume that $(G,\sigma,k_{vd},k_{ed},k_{ea})$ is a yes-instance of \textsc{Extended Editing to a Graph with a Given Degree Sequence}. Let $(U,D,A)$ be a solution. We show that if $d_G(v)> \delta+k_{vd}+k_{ed}$, then $v\in U$. To obtain a contradiction, assume that $d_G(v)> \delta+k_{vd}+k_{ed}$ but $v\notin U$. Then $d_{G'}(v)\leq \Delta$, where $G'=G-U-D+A$. It remains to observe that to decrease the degree of $v$ by at least $k_{vd}+k_{ed}+1$, we need at least $k_{vd}+k_{ed}+1$ vertex or edge deletion operations; a contradiction. We conclude that if $(G,\sigma,k_{vd},k_{ed},k_{ea})$ is a yes-instance, then the instance obtained by the application of the rule is also a yes-instance. It is straightforward to see that if $(G',\sigma,k_{vd}',k_{ed},k_{ea})$ is a yes-instance of \textsc{Extended Editing to a Graph with a Given Degree Sequence} obtained by the deletion of  a vertex $v$ and $(U,D,A)$ is a solution, then $(U\cup\{v\},D,A)$ is a solution for the original instance. Hence, the rule is safe. 

\medskip
We exhaustively apply the rule until we either stop and return a NO-answer or obtain an instance of the problem such that the degree of any vertex $v$ is at most $\Delta+k$. To simplify notations, we assume that $(G,\sigma,k_{vd},k_{ed},k_{ea})$ is such an instance.

On the next stage of the algorithm we apply the random separation technique.
We color the vertices of $G$ independently and uniformly at random by three colors. In other words, we partition $V(G)$ into three sets $R_v$, $Y_v$ and $B_v$ (some sets could be empty), and say that the vertices of $R_v$ are \emph{red}, the vertices of $Y_v$ are \emph{yellow} 
and the vertices of $B_v$ are \emph{blue}. Then the edges of $G$ are colored by either \emph{red} or \emph{blue}. We denote by $R_e$ the set of red and by $B_e$ the set of blue edges respectively. 

We are looking for a solution  $(U,D,A)$ of $(G,S,k_{vd},k_{ed},k_{ea})$ such that the vertices of $U$ are colored red, the vertices incident to the edges of $A$ are yellow and the edges of $D$ are red. Moreover, if $X$ and $Y$ are the sets of vertices incident to the edges of $D$ and $A$ respectively, then 
the vertices of $(N_G^2[U]\cup N_G[X\cup Y])\setminus (U\cup Y)$ and the edges of $E(G)\setminus D$ incident to the vertices of $N_G[U]\cup X\cup Y$ should be blue.
Formally, 
we say that a solution $(U,D,A)$ of $(G,S,k_{vd},k_{ed},k_{ea})$ is a \emph{colorful} solution if there are $R_v^*\subseteq R_v$, $Y_v^*\subseteq Y_v$ and $R_e^*\subseteq R_e$ such that the following holds.
\begin{itemize}
\item[i)] $|R_v^*|\leq k_{vd}$, $|R_e^*|\leq k_{ed}$ and $|Y_v^*|\leq 2k_{ea}$.
\item[ii)] $U=R_v^*$, $D=R_e^*$,  and for any $uv\in A$, $u,v\in Y_v^*$ and $|A|\leq k_{ea}$.
\item[iii)] If $u,v\in R_v\cup Y_v$  and $uv\in E(G)$, then either $u,v\in R_v^*\cup Y_v^*$ or $u,v\notin R_v^*\cup Y_v^*$.
\item[iv)] If $u\in  R_v\cup Y_v$ and $uv\in R_e$, then either $u\in R_v^*\cup Y_v^*,uv\in R_e^*$ or $u\notin R_v^*\cup Y_v^*,uv\notin R_e^*$.
\item[v)] If $uv,vw\in R_e$, then either $uv,vw\in R_e^*$ or $uv,vw\notin R_e^*$.
\item[vi)] If distinct $u,v\in R_v$ and $N_G(u)\cap N_G(v)\neq\emptyset$, then  either $u,v\in R_v^*$ or $u,v\notin R_v^*$.
\item[vii)] If $u\in R_v$ and $vw\in R_e$ for $v\in N_G(u)$, then either $u\in R_v^*,vw\in R_e^*$ or $u\notin R_v^*,vw\notin R_e^*$.
\end{itemize}
We also say that $(R_v^*,Y_v^*,R_e^*)$ is the \emph{base} of $(U,D,A)$. 

Our aim is to find a colorful solution if it exists. We do is by a dynamic programming algorithm based of the following properties of colorful solutions.

Let 
$$L=R_e\cup\{e\in E(G)\mid e\text{ is incident to a vertex of }R_v\}\cup\{uv\in E(G)\mid u,v\in Y_v\},$$
and $H=G[L]$.
Denote by $H_1,\ldots,H_s$ the components of $H$. Let $R_v^i=V(H_i)\cap R_e$, $Y_v^i=V(H_i)\cap Y_v$ and $R_e^i=E(H_i)\cap R_e$ for $i\in\{1,\ldots,s\}$.

\begin{myclaim}
\label{myclaim-A}
If $(U,D,A)$ is a colorful solution and $(R_v^*,Y_v^*,R_e^*)$ is its base, then
if  $H_i$ has a vertex of $R_v^*\cup Y_v*$ or en edge of $R_e^*$, then $R_v^i\subseteq R_v^*$, $Y_v^i\subseteq Y_v^*$ and $R_e^i\subseteq R_r^*$ for $i\in\{1,\ldots,s\}$.
\end{myclaim}

\begin{proofofclaim}[Proof of Claim \protect\ref{myclaim-A}]
Suppose that $H_i$ has $u\in R_v^*\cup Y_v^*$ or $e\in R_e^*$. 

If $v\in R_v^i\cup Y_v^i$, then $H_i$ has a path $P=x_0\ldots x_\ell$ such that 
$u=x_0$ or $e=x_0x_1$, and $x_\ell=v$.  By induction on $\ell$, we show that $v\in R_v^*$ or $v\in Y_v^*$ respectively. If $\ell=1$, then the statement follows from iii) and iv) of the definition of a colorful solution. Suppose that $\ell>1$. We consider three cases.

\medskip
\noindent
{\bf Case~1.}  $x_1\in R_v\cup Y_v$. By  iii) and iv), $x_1\in R_v^*\cup Y_v^*$ and, because the $(x_1,x_\ell)$-subpath of $P$ has length $\ell-1$, we conclude that  $v\in R_v^*$ or $v\in 
Y_v^*$ by induction. 

\medskip
Assume from now that $x_1\notin  R_v\cup Y_v$. 

\medskip
\noindent
{\bf Case~2.} $x_0x_1\in R_e$. Clearly, if for the first edge $e$ of $P$, $e\in R_e^*$, then $x_0x_1=e\in R_e^*$. Suppose that for the first vertex $u=x_0$ of $P$, $u\in R_v^*\cup Y_v^*$. Then 
by iv), $x_0x_1\in R_e^*$. If $x_1x_2\in R_e$, then $x_1x_2\in R_e^*$ by v). Since $x_1x_2\in R_e^*$ and  the $(x_1,x_\ell)$-subpath of $P$ has length $\ell-1$, we have that  $v\in R_v^*$ or $v\in Y_v^*$ by induction. Suppose that $x_1x_2\notin R_e$. Then because $x_1x_2\in L$, $x_2\in R_v$ and by vii), $x_2\in R_v^*$.  If $\ell=2$, then $x_\ell\in R_v^*$. Otherwise, as the 
 $(x_2,x_\ell)$-subpath of $P$ has length $\ell-2$, we have that  $v\in R_v^*$ or $v\in Y_v^*$ by induction.

\medskip
\noindent
{\bf Case~2.} $x_0x_1\notin R_e$. Then $u=x_0\in R_v^*\cup Y_v^*$. Because $x_0x_1\in L$, $x_0\in R_v^*$. 
If $x_1x_2\in R_e$, then $x_1x_2\in R_e^*$ by vii). Since $x_1x_2\in R_e^*$ and  the $(x_1,x_\ell)$-subpath of $P$ has length $\ell-1$, we have that  $v\in R_v^*$ or $v\in Y_v^*$ by induction. Suppose that $x_1x_2\notin R_e$. Then because $x_1x_2\in L$, $x_2\in R_v$ and by vi), $x_2\in R_v^*$.  If $\ell=2$, then $x_\ell\in R_v^*$. Otherwise, as the 
 $(x_2,x_\ell)$-subpath of $P$ has length $\ell-2$, we have that  $v\in R_v^*$ or $v\in Y_v^*$ by induction.
\medskip

Suppose that $e'\in R_e^i$. Then  $H_i$ has a path $P=x_0\ldots x_\ell$ such that 
$u=x_0$ or $e=x_0x_1$, and $x_{\ell-1}x_\ell=e'$.  Using the same inductive arguments as before, we obtain that $e'\in R_e^*$.
\end{proofofclaim}

\medskip

By Claim~\ref{myclaim-A}, we have that if there is a colorful solution $(U,D,A)$, then for its base $(R_v^*,Y_v^*,R_e^*)$, $R_v^*=\cup_{i\in I}R_v^i$,  $Y_v^*=\cup_{i\in I}Y_v^i$ and $R_e^*=\cup_{i\in I}R_e^i$
for some set of indices $I\subseteq\{1,\ldots,s\}$.  

The next property is a straightforward corollary of the definition $H$.

\begin{myclaim}
\label{myclaim-B}
For distinct $i,j\in\{1,\ldots,s\}$, if $u\in V(H_i)$ and $v\in V(H_j)$ are adjacent in $G$, then either $u,v\in B_v$ or 
($u\in Y_v^i$ and $v\in B_v$) or ($u\in B_v$ and $v\in Y_v^j$).
\end{myclaim}

We construct a dynamic programming algorithm that consecutively for $i=0,\ldots,s$, constructs the table $T_i$
that 
contains the records of values of the function $\gamma$:
$$\gamma(t_{vd},t_{ed},t_{ea},X,\delta)=(U,D,A,I),$$
where 
\begin{itemize}
\item[i)] $t_{vd}\leq k_{vd}$, $t_{ed}\leq k_{ed}$ and $t_{ea}\leq k_{ea}$,
\item[ii)] $X=\{d_1,\ldots,d_h\}$ is a collection (multiset) of integers, where $h\in\{1,\ldots,2t_{ea}\}$ and $d_i\in \{0,\ldots,\Delta\}$ for $i\in\{1,\ldots,h\}$,
\item[iii)] $\delta=(\delta_0,\ldots,\delta_{r})$, where $r=\max\{\Delta,\Delta(G)\}$ and
$\delta_i$ is a nonnegative integer 
for $i\in\{0,\ldots,r\}$,
\end{itemize}
such that $(U,D,A)$ is a \emph{partial solution} with the base $(R_v^*,Y_v^*,R_e^*)$ defined by $I\subseteq\{1,\ldots,i\}$ with the following properties.
\begin{itemize}
\item[iv)] 
$R_v^*=\cup_{i\in I}R_v^i$,  $Y_v^*=\cup_{i\in I}Y_v^i$ and $R_e^*=\cup_{i\in I}R_e^i$, and $t_{vd}=|R_v^*|$ and
$t_{ed}=|R_e^*|$.
\item[v)]  $U=R_v^*$, $D=R_e^*$,  $|A|=t_{ea}$ and for any $uv\in A$, $u,v\in Y_v^*$.
\item[vi)] The multiset $\{d_{G'}(y)\mid y\in Y_v^* \}=X$, where $G'=G-U-D+A$.
\item[vii)] $\delta(G')=\delta$.
\end{itemize}
In other words, $t_{vd}$ ,$t_{ed}$ and $t_{ea}$ are the numbers of deleted vertices, deleted edges and added edges respectively, $X$ is the multiset of degrees of of yellow vertices in the base of a partial solution, and $\delta$ is the degree vector of the graph obtained from $G$ by the editing with respect to a partial solution.
Notice that the values of $\gamma$ are defined only for some $t_{vd},t_{ed},t_{ea},X,\delta$ that satisfy i)--iii), as a partial solution with the properties iv)--vii) not necessarily exists, and we only keep records corresponding to the arguments $t_{vd},t_{ed},t_{ea},X,\delta$ for which $\gamma$ is defined.  

Now we explain how we construct the tables for $i\in\{0,\ldots,s\}$. 

\medskip
\noindent
{\bf Construction of $T_0$.} The table $T_0$ contains the unique record $(0,0,0,\emptyset,\delta)=(\emptyset,\emptyset,\emptyset,\emptyset)$, where 
$\delta=\delta(G)$ (notice that the length of $\delta$ can be bigger that the length of $\delta(G)$).

\medskip
\noindent
{\bf Construction of $T_i$ for $i\geq 1$.} We assume that $T_{i-1}$ is already constructed. Initially we set $T_i=T_{i-1}$. Then for each record  $\gamma(t_{vd},t_{ed},t_{ea},X,\delta)=(U,D,A,I)$ in $T_{i-1}$, we construct new records $\gamma(t_{vd}',t_{ed}',t_{ea}',X',\delta')=(U',D',A')$ and put them in $T_i$ unless $T_i$ already contains the value $\gamma(t_{vd}',t_{ed}',t_{ea}',X',\delta')$. In the last case we keep the old value.

Let $(t_{vd},t_{ed},t_{ea},X,\delta)=(U,D,A,I)$ in $T_{i-1}$. 
\begin{itemize}
\item If $t_{vd}+|R_v^i|>k_{vd}$ or $t_{ed}+|R_e^i|>k_{ed}$ or $t_{ea}+2|Y_v^i|>k_{ea}$, then stop considering the record. Otherwise,
let $t_{vd}'=t_{vd}+|R_v^i|$ and $t_{ed}'=t_{ed}+|R_e^i|$.
\item Let $F=G-U-D+A-R_v^i-R_e^i$. 
\item Let $\cup_{j\in I}Y_v^j=\{x_1,\ldots,x_h\}$, $d_F(x_f)=d_f$ for $f\in\{1,\ldots,h\}$. Let
$Y_v^i=\{y_1,\ldots,y_{\ell}\}$. Consider every $E_1\subseteq \binom{Y_v^i}{2}\setminus E(F[Y_v^i])$ and $E_2\subseteq\{x_fy_i\mid 1\leq f\leq h,1\leq j\leq\ell\}$ such that $|E_1|+|E_2|\leq k_{ea}-t_{ea}$, and set 
$\alpha_f=|\{x_fy_j\mid x_fy_j\in E_2,1\leq j\leq\ell\}|$ for $f\in\{1,\ldots,h\}$ and set 
$\beta_j=|\{e\mid e\in E_1,e\text{ is incident to }y_j\}|+|\{x_fy_j\mid x_fy_j\in E_2,1\leq f\leq h\}|$ for $j\in\{1,\ldots,\ell\}$.
\begin{itemize}
\item If $d_f+\alpha_f>\Delta$ for some $f\in\{1,\ldots,h\}$ or $d_{F}(y_j)+\beta_j>\Delta$ for some $j\in \{1,\ldots,\ell\}$, then stop considering the pair $(E_1,E_2)$.
\item Set $t_{ea}'=t_{ea}+|E_1|+|E_2|$, $X'=\{d_1+\alpha_1,\ldots,d_h+\alpha_h,d_F(y_1)+\beta_1,\ldots,d_F(y_\ell)+\beta_\ell\}$.
\item Let $F'=F+E_1+E_2$. Construct $\delta'=(\delta_0',\ldots,\delta_{r}')=\delta(F')$. 
\item Set $U'=U\cup R_v^i$, $D'=D\cup R_e^i$, $A'=A\cup E_1\cup E_2$, $I'=I\cup\{i\}$, set $\gamma(t_{vd}',t_{ed}',t_{ea}',X',\delta')=(U',D',A',I')$ and put the record in $T_i$.
\end{itemize}
\end{itemize}

We consecutively construct $T_1,\ldots,T_s$.  The algorithm returns a YES-answer if $T_s$ contains a record $(t_{vd},t_{ed},t_{ea},X,\delta)=(U,D,A,I)$ for $\delta=\delta(\sigma)$ and $(U,D,A)$ is a colorful solution in this case. Otherwise, the algorithm returns a NO-answer.

The correctness of the algorithm follows from the next claim.

\begin{myclaim}
\label{myclaim-C}
For each $i\in\{1,\ldots,s\}$, the table $T_i$ contains a record $\gamma(t_{vd},t_{ed},t_{ea},X,\delta)=(U,D,A,I)$, if and only if 
there are $t_{vd},t_{ed},t_{ea},X,\delta$ satisfying i)-iii) such that there is a partial solution $(U^*,D^*,A^*)$ and $I^*\subseteq\{1,\ldots,i\}$ that satisfy iv)-vii). In particular $t_{vd},t_{ed},t_{ea},X,\delta$,  $(U,D,A)$ and $I$ satisfy i)--vii) if $\gamma(t_{vd},t_{ed},t_{ea},X,\delta)=(U,D,A,I)$ is in $T_i$.
\end{myclaim}

\begin{proofofclaim}[Proof of Claim~\ref{myclaim-C}]
We prove the claim by induction on $i$. It is straightforward to see that it holds for $i=0$. Assume that $i>0$ and the claim is fulfilled for $T_{i-1}$.

Suppose that a record $\gamma(t_{vd}',t_{ed}',t_{ea}',X',\delta')=(U',D',A',I')$ was added in $T_i$. Then ether  $\gamma(t_{vd}',t_{ed}',t_{ea}',X',\delta')=(U',D',A',I')$ was in $T_{i-1}$ or it was constructed for some record $(t_{vd},t_{ed},t_{ea},X,\delta)=(U,D,A,I)$ from $T_{i-1}$. In the first case, 
$t_{vd}',t_{ed}',t_{ea}',X',Q'$, $(U',D',A')$ and $I'\subseteq\{1,\ldots,i\}$  satisfy i)-vii) by induction. Assume that $\gamma(t_{vd}',t_{ed}',t_{ea}',X',\delta')=(U',D',A',I')$ was constructed for some record $(t_{vd},t_{ed},t_{ea},X,Q)=(U,D,A,I)$ from $T_{i-1}$. Notice that $i\in I'$ in this case. Let $I=I'\setminus \{i\}$.
Consider $\cup_{j\in I}Y_v^j=\{x_1,\ldots,x_h\}$ and $Y_v^i=\{y_1,\ldots,y_{\ell}\}$. By Claim~\ref{myclaim-B}, $x_f$ and $y_j$ are not adjacent for $f\in\{1,\ldots,h\}$ and $j\in\{1,\ldots,\ell\}$. Then it immediately follows from the description of the algorithm that $t_{vd}',t_{ed}',t_{ea}',X',\delta'$, 
$(U',D',A')$ and $I'$ satisfy i)--vii). 

Suppose that there are $t_{vd},t_{ed},t_{ea},X,\delta$ satisfying i)-iii) such that there is a partial solution $(U^*,D^*,A^*)$ and $I^*\subseteq\{1,\ldots,i\}$ that satisfy iv)-vii).
Suppose that $i\notin I^*$. Then $T_{i-1}$ contains a record  $\gamma(t_{vd},t_{ed},t_{ea},X,\delta)=(U,D,A,I)$ by induction and, therefore, this record is in $T_i$. Assume from now that $i\in I^*$.
Let $I'=I^*\setminus \{i\}$. Consider $R_v'=\cup_{j\in I'}R_v^j$ and $Y_v'=\cup_{j\in I'}Y_v^j$.
Let $E_1=\{uv\in A\mid u,v\in T_v^i\}$ and $E_2=\{uv\in A\mid u\in Y_v',v\in Y_v^i\}$.
Define $U'=U\setminus R_v^i$, $D'=D\setminus R_e^i$ and $A'=A\setminus(E_1\cup E_2)$.
Let $t_{vd}'=|U'|$, $t_{ed}=|D'|$ and $t_{ea}=|A'|$.
Consider the multiset of integers $X'=\{d_F(v)\mid v\in Y_v'\}$ and the sequence $\delta'=(\delta_1',\ldots,\delta_r')=\delta(F)$ for $F=G-U'-D'+A'$.
We obtain that $t_{vd}',t_{ed}',t_{ea}',X',\delta'$, $(U',D',A')$ and $I'\subseteq\{1,\ldots,i-1\}$  satisfy i)-vii). 
By induction, $T_{i-1}$ contains a record $\gamma(t_{vd}',t_{ed}',t_{ea}',X',\delta')=(U'',D'',A'',I'')$.
Let $Y_v'=\{x_1,\ldots,x_h\}$, $\cup_{j\in I''}Y_v^j=\{x_1',\ldots,x_h'\}$ and assume that $d_F(x_f)=d_{F'}(x_f')$ for $f\in\{1,\ldots,h\}$, where $F'=G-U''-D''+A''$. 
Consider $E_2'$ obtained from $E_2$ by the replacement of every edge $x_fv$ by $x_f'v$ for $f\in\{1,\ldots,h\}$ and $v\in Y_v^i$. It remains to observe that when we consider $\gamma(t_{vd}',t_{ed}',t_{ea}',X',\delta')=(U'',D'',A'',I'')$ and the pair $(E_1,E_2')$, we obtain $\gamma(t_{vd},t_{ed},t_{ea},X,\delta)=(U,D,A,I)$ for
$U=U''\cup R_v^i$, $D=D''\cup R_e^i$, $A=A''\cup E_1\cup E_2'$ and $I=I''\cup\{i\}$.
\end{proofofclaim}

\medskip

Now we evaluate the running time of the dynamic programming algorithm. 

First, we upper bound the size of each table. Suppose that $\gamma(t_{vd},t_{ed},t_{ea},X,\delta)=(U,D,A,I)$ is included in a table $T_i$. 
By the definition and Claim~\ref{myclaim-C}, $\delta=\delta(G')$ for $G'=G-U-D+A$. Let $\delta=\{\delta_0,\ldots,\delta_r\}$ and $\delta(G)=(\delta_0',\ldots,\delta_r')$.   Let $i\in\{0,\ldots,r\}$. Denote $W_i=\{v\in V(G)\mid d_G(v)=i\}$. Recall that $\delta(G)\leq \Delta+k$. 
 If $\delta_i'>\delta_i$, then at least $\delta_i'-\delta_i$ vertices of $W_i$ should be either deleted or get modified degrees by the editing with respect to $(U,D,A)$. Since at most $k_{vd}$ vertices of $W_i$ can be deleted and we can modify degrees of at most $(k+\Delta)k_{vd}+2(k_{ed}+k_{ea})$ vertices, 
$\delta_i'-\delta_i\leq (k+\Delta+1)k_{vd}+2(k_{ed}+k_{ea})$.
Similarly, if $\delta_i>\delta_i'$, then at least $\delta_i-\delta_i'$ vertices of $V(G)\setminus W_i$ should get modified degrees. Since we can modify degrees of at most $(k+\Delta)k_{vd}+2(k_{ed}+k_{ea})$ vertices, $\delta_i-\delta_i'\leq (k+\Delta)k_{vd}+2(k_{ed}+k_{ea})$. We conclude
that for each $i\in\{0,\ldots,r\}$, 
$$\delta_i'-(k+\Delta+1)k_{vd}+2(k_{ed}+k_{ea})\leq  \delta_i\leq \delta_i'+(k+\Delta)k_{vd}+2(k_{ed}+k_{ea})$$
and, therefore, there are at most 
$(2(k+\Delta)k_{vd}+4(k_{ed}+k_{ea})+1)^r$ distinct vectors $\delta$.  Since $r=\max\{\Delta,\Delta(G)\}\leq \Delta+k$, we have 
$2^{O((\Delta+k)\log (\Delta+k))}$ distinct vectors $\delta$. The number of distinct multisets $X$ is at most $(\Delta+1)^{2k}$ and there are at most $3(k+1)$ possibilities for $t_{vd},t_{ed},t_{ea}$. We conclude that each $T_i$ has  $2^{O((\Delta+k)\log (\Delta+k))}$ records.

To construct a new record $\gamma(t_{vd}',t_{ed}',t_{ea}',X',\delta')=(U',D',A',I')$ from  $\gamma(t_{vd},t_{ed},t_{ea},X,\delta)=(U,D,A,I)$ we consider 
all possible choices of $E_1$ and $E_2$. Since these edges have their end-vertices in a set of size at most $2k_{ea}$ and $|E_1|+|E_2|\leq k_{ea}$, there are $2^{O(k\log k)}$ possibilities to choose $E_1$ and $E_2$. The other computations in the construction of  $\gamma(t_{vd}',t_{ed}',t_{ea}',X',\delta')=(U',D',A',I')$  can be done  in linear time. We have that $T_i$ can be constructed from $T_{i-1}$ in time $2^{O((\Delta+k)\log (\Delta+k))}\cdot n$ for $i\in\{1,\ldots,s\}$. Since $s\leq n$, the total time is $2^{O((\Delta+k)\log (\Delta+k))}\cdot n^2$. 

\medskip
We proved that a colorful solution can be found in time $2^{O((\Delta+k)\log (\Delta+k))}\cdot n^2$ if exist. Clearly, any colorful solution is a solution for $(G,\sigma,k_{vd},k_{ed},k_{ea})$ and we can return it, but nonexistence of a colorful solution does not imply that there is no solution. Hence, to find a solution, we run the randomized algorithm  $N$ times, i.e., we consider $N$ random colorings and try to find a colorful solution for them. If we find a solution after some run, we return it and stop. If we do not obtain a solution after $N$ runs, we return a NO-answer. The next claim shows that it is sufficient to run the algorithm $N=6^{2k(\Delta+k)^2}$ times.

\begin{myclaim}
\label{myclaim-D}
There is a positive $p$ that does not depend on the instance  such that 
if after $N=6^{2k(\Delta+k)^2}$ executions the randomized algorithm does not find a solution for $(G,\sigma,k_{vd},k_{ed},k_{ea})$, then 
the probability that  $(G,\sigma,k_{vd},k_{ed},k_{ea})$ is a no-instance is at least $p$.
\end{myclaim}

\begin{proofofclaim}[Proof of Claim~\ref{myclaim-D}]
Suppose that $(G,\sigma,k_{vd},k_{ed},k_{ea})$ has a solution $(U,D,A)$.  
Let $X$ be the set of end-vertices of the edges of $D$ and $Y$ is the set of end-vertices of $A$. 
Let $W=N_G^2[U]\cup N_G[X\cup Y]$ and 
denote by $L$ the set of edges incident to the vertices of $N_G[U]\cup X\cup Y$.
The algorithm colors the vertices of $G$ independently and uniformly at random by three colors and the edges are colored by two colors. 
Notice that if the vertices of $W$ and the edges of $L$ are colored correctly with respect to the solution, i.e., the vertices of $U$ are red, the vertices of $Y$ are yellow, all the other vertices are blue, the edges of $D$ are red and all the other edges are blue, then $(U,D,A)$ is a colorful solution. Hence, the algorithm can find a solution in this case.

We find a lower bound for the probability that the vertices of $W$ and the edges of $L$ are colored correctly with respect to the solution.
Recall that $\Delta(G)\leq \Delta+k$. 
Hence, $|W|\leq k_{vd}(\Delta+k)^2+2(k_{ed}+k_{ea})(\Delta+k)\leq 2k(\Delta+k)^2$ and $|L|\leq  k_{vd}(\Delta+k)^2+2(k_{ed}+k_{ea})(\Delta+k)\leq 2k(\Delta+k)^2$. As the vertices are colored by three colors and the edges by two, we obtain that the probability that the vertices of $W$ and the edges of $L$ are colored correctly with respect to the solution is at least 
$3^{-2k(\Delta+k)^2}\cdot 2^{-2k(\Delta+k)^2}=6^{-2k(\Delta+k)^2}$. 

The probability that the vertices of $W$ and the edges of $L$ are not colored correctly with respect to the solution is at most $1-6^{-2k(\Delta+k)^2}$, and the probability that these vertices are non colored correctly with respect to the solution for neither of $N=6^{2k(\Delta+k)^2}$ random colorings is at most   $(1-1/N)^{N}$, and the claim follows.
\end{proofofclaim}

\medskip
Claim~\ref{myclaim-D} implies that the running time of the randomized algorithm is $2^{O(k(\Delta+k)^2}\cdot n^2$.

The algorithm can be derandomized by standard techniques (see~ \cite{AlonYZ95,CaiCC06}) because random colorings can be replaced by the colorings induced by \emph{universal sets}.
Let $m$ and $r$ be positive integers, $r\leq m$. An  \emph{$(m,r)$-universal set} is a collection of binary vectors of length $m$ such that for each index subset of size $r$, each of the $2^r$ possible combinations of values appears in some vector of the set. It is known that an $(m,r)$-universal set can be constructed in \classFPT-time with the parameter $r$. The best construction is due to Naor, Schulman and Srinivasan~\cite{naor1995splitters}. They obtained an $(m,r)$-universal set of size $2^r\cdot r^{O(\log r)} \log m$, and proved that the elements of the sets  can be listed in time that is linear in the size of the set. 

In our case we have $m=|V(G)|+|E(G)|\leq ((\Delta+k)/2+1)n$ and $r=4k(\Delta+k)^2$, as we have to obtain the correct coloring of $W$ and $L$ corresponding to a solution $(U,D,A)$. 
Observe that colorings induced by a universal set are binary and we use three colors. To fix it, we assume that the coloring of the vertices and edges is done in two stages. First, we color the elements of $G$ by two colors: red and green, and then recolor the green elements by yellow or blue. By using  an $(m,r)$-universal set of size $2^r\cdot r^{O(\log r)} \log m$, we get 
 $4^r\cdot r^{O(\log r)} \log m$ colorings by three colors. We conclude that the running time of the derandomized algorithm is $2^{O(k(\Delta+k)^2}\cdot n^2\log n$.
\qed
\end{proof}

\section{Kernelization for  Editing to a Graph with a Given Degree Sequence}\label{sec:kernel}
In this section we show that \textsc{Editing to a Graph with a Given Degree Sequence} has a polynomial kernel when parameterized by $k+\Delta$ if $S=\{ea\}$, but for all other nonempty $S\subseteq \{vd,ed,ea\}$, there is no polynomial kernel unless  $\classNP\subseteq\classCoNP/\text{\rm poly}$.

\begin{theorem}\label{thm:kernel}
If $S=\{\ea\}$, then \textsc{Editing to a Graph with a Given Degree Sequence} parameterized by $k+\Delta$ has a kernel with $O(k\Delta^2)$ vertices, where $\Delta=\max \sigma$.  
\end{theorem}

\begin{proof}
Let $(G,\sigma,k)$ be an instance of \textsc{Editing to a Graph with a Given Degree Sequence} and $\Delta=\max \sigma$. If $\Delta(G)>\Delta$,  $(G,\sigma,k)$ is a no-instance, because by edge additions it is possible only increase degrees. Hence, we immediately stop and return a NO-answer in this case. Assume from now that $\Delta(G)\leq \Delta$.
For $i\in\{0,\ldots,\Delta\}$, denote $W_i=\{v\in V(G)\mid d_G(v)=i\}$ and $\delta_i=|W_i|$. Let $s_i=\min\{\delta_i,2k(\Delta+1)\}$ and let $W_i'\subseteq W_i$ be an arbitrary set of size $s_i$ for $i\in\{0,\ldots,\Delta\}$.  We consider $W=\cup_{i=0}^\Delta W_i'$ and prove the following claim.

\begin{myclaim}
\label{myclaim-E}
If $(G,\sigma,k)$ is a yes-instance of \textsc{Editing to a Graph with a Given Degree Sequence}, then there is 
$A\subseteq \binom{V(G)}{2}\setminus E(G)$ such that $\sigma(G+A)=\sigma$, $|A|\leq k$ and for any $uv\in A$, $u,v\in  W$.
\end{myclaim}

\begin{proofofclaim}[Proof of Claim~\ref{myclaim-E}]
Suppose that $A\subseteq \binom{V(G)}{2}\setminus E(G)$ is a solution for   $(G,\sigma,k)$, i.e., $\sigma(G+A)=\sigma$ and $|A|\leq k$, such that the total number of end-vertices of the edges of $A$ in $V(G)\setminus W$ is minimum. Suppose that there is $i\in\{0,\ldots,\Delta\}$ such that at least one edge of $A$ has its end-vertex in $W_i\setminus W_i'$. 
Clearly, $s_i= 2k(\Delta+1)$.
Denote by $\{x_1,\ldots,x_p\}$ the set of end-vertices of the edges of $A$ in $W_i$ and let $\{y_1,\ldots,y_q\}$ be the set of end-vertices of the edges of $A$ in $V(G)\setminus W_i$. 
Since $p+q\leq 2k$, $\Delta(G)\leq \Delta$ and $s_i=2k(\Delta+1)$, there is a set of vertices $\{x_1',\ldots,x_p'\}\subseteq W_i'$ such that the vertices of this set are pairwise nonadjacent and are not adjacent to the vertices of $\{y_1,\ldots,y_q\}$. We construct $A'\subseteq \binom{V(G)}{2}\setminus E(G)$ by replacing every edge $x_iy_j$ by $x_i'y_j$ for $i\in\{1,\ldots,p\}$ and $j\in\{1,\ldots,q\}$, and every edge $x_ix_j$ is replaced by $x_i'x_j'$ for $i,j\in\{1,\ldots,q\}$. It is straightforward to verify that $A'$ is a solution for  
 $(G,\sigma,k)$, but $A'$ has less end-vertices outside $W$ contradicting the choice of $A$. Hence, no edge of $A$ has an end-vertex in $V(G)\setminus W$.
\end{proofofclaim}

\medskip

If $\delta_i\leq 2k(\Delta+1)$ for $i\in\{0,\ldots,\Delta\}$, then we return the original instance $(G,\sigma,k)$ and stop, as $|V(G)|\leq 2k(\Delta+1)^2$. From now we assume that there is $i\in\{0,\ldots,\Delta\}$ such that $\delta_i>2k(\Delta+1)$.
We construct the graph $G'$ as follows.
\begin{itemize}
\item Delete all the vertices of $V(G)\setminus W$.
\item Construct $h=\Delta+2$ new vertices $v_1,\ldots,v_h$ and join them by edges pairwise to obtain a clique.  
\item For any $u\in W$ such that $r=|N_G(u)\cap (V(G)\setminus W)|\geq 1$, construct edges $uv_1,\ldots,uv_r$. 
\end{itemize}
Notice that $d_{G'}(v_1)\geq\ldots\geq d_{G'}(v_h)\geq\Delta+1$ and $d_{G'}(u)=d_G(u)$ for $u\in W$.
Now we consider the sequence $\sigma$ and construct the sequence $\sigma'$ as follows.
\begin{itemize}
\item The first $h$ elements of $\sigma'$ are $d_{G'}(v_1),\ldots,d_{G'}(v_h)$.
\item Consider the elements of $\sigma$ in their order and for each integer $i\in\{0,\ldots,\Delta\}$ that occurs $j_i$ times in $\sigma$, add $j_i-(\delta_i-s_i)$ copies of $i$ in $\sigma'$.
\end{itemize}
We claim that $(G,\sigma,k)$ is a yes-instance of \textsc{Editing to a Graph with a Given Degree Sequence} if and only if $(G',\sigma',k)$ is a yes-instance of the problem.

Suppose that $(G,\sigma,k)$ is a yes-instance of \textsc{Editing to a Graph with a Given Degree Sequence}. By Claim~\ref{myclaim-E}, it has a solution 
$A\subseteq \binom{V(G)}{2}\setminus E(G)$ such that  for any $uv\in A$, $u,v\in  W$. It is straightforward to verify that $\sigma(G'+A)=\sigma'$, i.e., $A$ is a solution for $(G',\sigma',k)$. Assume that $A\subseteq \binom{V(G')}{2}\setminus E(G)$ is a solution for $(G',\sigma',k)$. Because  $d_{G'}(v_1),\ldots,d_{G'}(v_h)$ are the first $h$ elements of $\sigma'$ and $d_{G'}(u)=d_G(u)\leq \Delta$ for $u\in W$, for any $uv\in A$, $u,v\in W$. Then  it is straightforward to check that $\sigma(G+A)=\sigma$, i.e.,  $A$ is a solution for $(G,\sigma,k)$.
\qed
\end{proof}

We complement Theorem~\ref{thm:kernel} by showing that it is unlikely that \textsc{Editing to a Graph with a Given Degree Sequence} parameterized by $k+\Delta$ has a polynomial kernel
for $S\neq\{\ea\}$. The proof is based on the cross-composition technique introduced by Bodlaender, Jansen and Kratsch~\cite{BodlaenderJK14}.

\begin{theorem}\label{thm:no-kernel}
If $S\subseteq\{\vd,\ed,\ea\}$ but $S\neq\{\ed\}$, then \textsc{Editing to a Graph with a Given Degree Sequence}  has no polynomial kernel unless  $\classNP\subseteq\classCoNP/\text{\rm poly}$ when the problem is parameterized by $k+\Delta$ for $\Delta=\max\sigma$.  
\end{theorem}

\begin{proof}
For the proof of the theorem, we  need some additional definitions and statements.
Recall that, formally, a parameterized problem $\mathcal{P}\subseteq\Sigma^*\times\mathbb{N}$, where $\Sigma$ is a finite alphabet.

Let $\Sigma$ be a finite alphabet. An equivalence relation $\mathcal{R}$ on the set of strings $\Sigma^*$ is called a \emph{polynomial equivalence relation} if the following two conditions hold:
\begin{itemize}
\item[i)] there is an algorithm that given two strings $x,y\in\Sigma^*$ decides whether $x$ and $y$ belong to
the same equivalence class in time polynomial in $|x|+|y|$,
\item[ii)] for any finite set $S\subseteq\Sigma^*$, the equivalence relation $\mathcal{R}$ partitions the elements of $S$ into a
number of classes that is polynomially bounded in the size of the largest element of $S$.
\end{itemize}

Let $L\subseteq\Sigma^*$ be a language, let $\mathcal{R}$ be a polynomial
equivalence relation on $\Sigma^*$, and let $\mathcal{P}\subseteq\Sigma^*\times\mathbb{N}$   
be a parameterized problem.  An \emph{OR-cross-composition of $L$ into $\mathcal{P}$} (with respect to $\mathcal{R}$) is an algorithm that, given $t$ instances $x_1,x_2,\ldots,x_t\in\Sigma^*$ 
of $L$ belonging to the same equivalence class of $\mathcal{R}$, takes time polynomial in
$\sum_{i=1}^t|x_i|$ and outputs an instance $(y,k)\in \Sigma^*\times \mathbb{N}$ such that:
\begin{itemize}
\item[i)] the parameter value $k$ is polynomially bounded in $\max\{|x_1|,\ldots,|x_t|\} + \log t$,
\item[ii)] the instance $(y,k)$ is a yes-instance for $\mathcal{P}$ if and only if at least one instance $x_i$ is a yes-instance for $L$ for $i\in\{1,\ldots,t\}$.
\end{itemize}
It is said that $L$ \emph{OR-cross-composes into} $\mathcal{P}$ if a cross-composition
algorithm exists for a suitable relation $\mathcal{R}$.

In particular, Bodlaender, Jansen and Kratsch~\cite{BodlaenderJK14} proved that 
if an \classNP-hard language $L$ OR-cross-composes into the parameterized problem $\mathcal{P}$,
then $\mathcal{P}$ does not admit a polynomial kernelization unless
$\classNP\subseteq\classCoNP/\text{\rm poly}$.

We prove that the \textsc{Clique} problem that asks for a graph $G$ and a positive integer $k$, whether $G$ has a clique of size $k$, OR-cross-composes into \textsc{Editing to a Graph with a Given Degree Sequence}. Recall that this problem is  \classNP-complete~\cite{GareyJ79} for regular graphs.

Suppose that $\ed\in S$.
We assume that two instances $(G,k)$ and $(G',k')$ of \textsc{Clique} are equivalent if $|V(G)|=|V(G')|$, $k=k'$ and $G,G'$ are $d$-regular for some nonnegative integer $d$.
 Let $(G_1,k),\ldots,(G_t,k)$ be equivalent instances of \textsc{Clique}, where $G_1,\ldots,G_t$ are $d$-regular, $n=|V(G_1)|=\ldots=|V(G_t)|$ and $d\geq k-1$. 
We construct the graph $G$ by taking the disjoint union of  copies of $G_1,\ldots,G_t$. 
Consider the sequence $\sigma=(\sigma_1,\ldots,\sigma_{nt})$, where
$$
\sigma_i=
\begin{cases}
d&\mbox{if }1\leq i\leq nt-k,\\
d-(k-1)&\mbox{if } nt-k+1\leq i\leq nt.
\end{cases}
$$
Let $k'=k(k-1)/2$. We claim that $(G_i,k)$ is a yes-instance of \textsc{Clique} for some $i\in\{1,\ldots,t\}$
if and only if $(G,\sigma,k')$ is a yes-instance of \textsc{Editing to a Graph with a Given Degree Sequence}.
If $K$ is a clique of size $k$ in $G_i$, then the graph $G'$ obtained  from $G$ by the deletion of the $k'=k(k-1)/2$ edges of $D=E(G[K])$ has the degree sequence $\sigma$. Assume that 
$(U,D,A)$ is a solution of $(G,\sigma,k)$. Clearly, $U=\emptyset$ even if $\vd\in R$, because $\sigma$ contains $nt$ elements. Since $\sum_{i=1}^{nt}\sigma_i=dn-k(k-1)$, we have that $A=\emptyset$. It remains to notice that because in $G-D$ $k$ vertices have degree $d-(k-1)$, $G[D]$ is a compete graph with $k$ vertices, i.e., $G$ contains a clique of size $k$. Clearly, any clique $K$ of size $k$ is a clique of some $G_i$ for $i\in\{1,\ldots,t\}$.

Assume that $S=\{\vd\}$. Now we assume that two instances $(G,k)$ and $(G',k')$ of \textsc{Clique} are equivalent if $|V(G)|=|V(G')|$, $|E(G)|=|E(G')|$, $k=k'$ and $G,G'$ are $d$-regular for some nonnegative integer $d$.
Let $(G_1,k),\ldots,(G_t,k)$ be equivalent instances of \textsc{Clique}, where $G_1,\ldots,G_t$ are $d$-regular, $n=|V(G_1)|=\ldots=|V(G_t)|$, $m=|E(G_1)|=\ldots=|E(G_t)|$ and $d-(k-1)\geq 3$. 
We construct the graph $G$ as follows.
\begin{itemize}
\item Take the disjoint union of copies of $G_1,\ldots,G_t$.
\item For each edge $uv\in E(G_i)$ for $i\in\{1,\ldots,t\}$, subdivide it, i.e., construct a new vertex $w$ and replace $uv$ by $uw$ and $wv$. We call the new vertices \emph{subdivision} vertices.
\end{itemize}
Let $k'=k(k-1)/2$.
Consider the sequence $\sigma=(\sigma_1,\ldots,\sigma_p)$, where $p=(n+m)t-k'$ and 
$$
\sigma_i=
\begin{cases}
d&\mbox{if } 1\leq i\leq nt-k,\\
d-(k-1)&\mbox{if } nt-k+1\leq i\leq nt,\\
2&\mbox{if } nt+1\leq i\leq p.
\end{cases}
$$
 We claim that $(G_i,k)$ is a yes-instance of \textsc{Clique} for some $i\in\{1,\ldots,t\}$
if and only if $(G,\sigma,k')$ is a yes-instance of \textsc{Editing to a Graph with a Given Degree Sequence}.
If $K$ is a clique of size $k$ in $G_i$, then the graph $G'$ obtained  from $G$ by the deletion of the $k'=k(k-1)/2$ subdivision vertices corresponding to the edges $G[K]$ has the degree sequence $\sigma$. Assume that 
$(U,D,A)$ is a solution of $(G,\sigma,k)$. Because $\sigma$ has $p$ elements and $|V(G)|-p=t(n+m)-p=k'$, $U$ contains $k'$ vertices and $D=A=\emptyset$. By the construction of $G$ and $\sigma$, $U$ contains only vertices of degree 2. As $d-(k-1)\geq 3$, we have that $U$ contains $k'$ subdivision vertices.
It remains to notice that because in $G-U$ $k$ vertices have degree $d-(k-1)$, the subdivision vertices of $U$ correspond to the edges of a compete graph with $k$ vertices, i.e., $G$ contains a clique of size $k$. Clearly, any clique $K$ of size $k$ is a clique of some $G_i$ for $i\in\{1,\ldots,t\}$.
\qed
\end{proof}


\end{document}